\documentclass{tlp}

\usepackage{aopmath}

\newtheorem{example}{Example} 
\newtheorem{definition}{Definition} 
\usepackage{color}

\newcommand\bcmdtab{\noindent\bgroup\tabcolsep=0pt%
  \begin{tabular}{@{}p{10pc}@{}p{20pc}@{}}}
\newcommand\ecmdtab{\end{tabular}\egroup}

\usepackage{amssymb}
\usepackage{marvosym}
\usepackage{stmaryrd}
\usepackage{proof}
\usepackage{url}

\usepackage{xspace}
\usepackage[all]{xy}

\renewcommand{\tt}{\ttfamily}
\newcommand{\codefont}{\small\tt}
\newcommand{\code}[1]{\mbox{\codefont{#1}}}
\newcommand{\ccode}[1]{``\code{#1}''}
\newcommand{\us}{\raise-.8ex\hbox{-}}
\newcommand{\xtilde}{\!\raise-.75ex\hbox{\char`\~}} 
\newcommand{\funset}{\ensuremath{_{\cal S}}} 
\newcommand{\Rc}{{\cal{R}}} 
\newcommand{\tostar}{\buildrel * \over \to}
\renewcommand{\emptyset}{\varnothing}
\newcommand{\ol}[1]{\overline{#1}}  
\newcommand{\tree}{{\cal T}}
\newcommand{\pattern}{\pi}

\usepackage{listings}
\lstnewenvironment{curry}{\lstset{literate={->}{{$\rightarrow{}\!\!\!$}}3}}{}
\lstnewenvironment{prolog}{}{}
\newcommand{\listline}{\vrule width0pt depth1.5ex}

\title[Default Rules for Curry]
      {Default Rules for Curry\footnote{%
This is an extended version of a paper presented at the
international symposium on Practical Aspects of Declarative Languages
(PADL 2016), invited as a rapid communication in TPLP.
The authors acknowledge the assistance of the conference program chairs
Marco Gavanelli and John Reppy.}}

\author[Sergio Antoy and Michael Hanus]
       {SERGIO ANTOY \\
         Computer Science Dept., Portland State University, Oregon, U.S.A.\\
         \email{antoy@cs.pdx.edu}
\and
        MICHAEL HANUS \\
         Institut f\"ur Informatik, CAU Kiel, D-24098 Kiel, Germany. \\
         \email{mh@informatik.uni-kiel.de}}

\submitted{January 31, 2016}
\revised{April 13, 2016}
\accepted{May 2, 2016}

\begin{document}

\label{firstpage}

\maketitle

\begin{abstract}
In functional logic programs,
rules are applicable independently of textual order, 
i.e., any rule can potentially be used to evaluate an expression.
This is similar to logic languages and contrary
to functional languages, e.g., Haskell
enforces a strict sequential interpretation of rules.
However, in some situations it is convenient
to express alternatives by means of compact default rules.
Although default rules are often used
in functional programs, the non-deterministic
nature of functional logic programs does not allow
to directly transfer this concept from functional
to functional logic languages in a meaningful way.
In this paper we propose a new concept of default rules
for Curry that supports a programming style similar
to functional programming while preserving the core
properties of functional logic programming, i.e.,
completeness, non-determinism, and logic-oriented use
of functions. We discuss the basic concept
and propose an implementation
which exploits advanced features of functional logic
languages.\\[2ex]
\emph{To appear in Theory and Practice of Logic Programming (TPLP)}
\end{abstract}


\section{Motivation}
\label{Motivation}

Functional logic languages combine the most important
features of functional and logic programming in a single language
(see \cite{AntoyHanus10CACM,Hanus13} for recent surveys).
In particular, the functional logic language Curry \cite{Hanus16Curry}
conceptually extends Haskell with common features of logic programming,
i.e., non-determinism, free variables, and constraint solving.
Moreover, the amalgamated features of Curry support
new programming techniques, like \emph{deep} pattern matching
through the use of \emph{functional patterns}, i.e., evaluable functions
at pattern positions \cite{AntoyHanus05LOPSTR}.

For example, suppose that we want to compute two elements $x$ and $y$
in a list $l$ with the
property that the distance between the two elements is $n$, i.e., in
$l$ there are $n-1$ elements between $x$ and $y$.  We will use this 
condition in the $n$-queens program discussed later.  Of course, there
may be many pairs of elements in a list satisfying the given condition
(\ccode{++} denotes the concatenation of lists):
\begin{curry}
dist n (_++[x]++zs++[y]++_) | n == length zs + 1 = (x,y)
\end{curry}
Defining functions by case distinction through pattern matching
is a very useful feature.  Functional patterns make this
feature even more convenient.
However, in functional logic languages, this feature
is slightly more delicate because of the possibility of
functional patterns, which typically stand for an infinite number of
standard patterns, and because there is no textual order among the rules
defining an operation. The variables in a functional pattern are bound
like the variables in ordinary patterns.

As a simple example, consider an operation \code{isSet}
intended to check whether a given list represents a set,
i.e., does not contain duplicates.
In Curry, we might think to implement it as follows:
\begin{curry}
isSet (_++[x]++_++[x]++_) = False
isSet _                   = True
\end{curry}
The first rule uses a functional pattern: it returns \code{False}
if the argument matches a list where two identical elements occur.
The intent of the second rule is to return \code{True}
if no identical elements occur in the argument.
However, according to the semantics of Curry,
which ensures completeness w.r.t.\ finding solutions or values,
\emph{all} rules are tried to evaluate an expression.
Therefore, the second rule is always applicable
to calls of \code{isSet} so that the expression
\code{isSet$\,$[1,1]} will be evaluated to \code{False} \emph{and}
\code{True}.

The unintended application of the second rule can be avoided
by the additional requirement that this rule should be applied only if
no other rule is applicable. We call such a rule a \emph{default rule}
and mark it by adding the suffix \code{'default} to the
function's name (in order to avoid a syntactic extension of the
base language).
Thus, if we define \code{isSet} with the rules
\begin{curry}
isSet (_++[x]++_++[x]++_) = False
isSet'default _           = True
\end{curry}
then \code{isSet$\,$[1,1]} evaluates only to \code{False}
and \code{isSet$\,$[0,1]} only to \code{True}.

In this paper we propose a concept for default rules for Curry,
define its precise semantics, and discuss implementation options.
In the next section, we review the main concepts of
functional logic programming and Curry.
Our intended concept of default rules is informally introduced
in Sect.~\ref{sec:concept}.
Some examples showing the convenience of default rules for programming
are presented in Sect.~\ref{sec:examples}.
In order to avoid the introduction of a new semantics
specific to default rules, we define the precise meaning
of default rules by transforming them into already known concepts
in Sect.~\ref{sec:transsem}.
Options to implement default rules efficiently are discussed
and evaluated in Sect.~\ref{sec:impl}.
Some benchmarking of alternative implementations of
default rules are shown in Sect.~\ref{sec:bench}
before we relate our proposal to other work and conclude.

\section{Functional Logic Programming and Curry}
\label{sec:flp}

Before presenting the concept and implementation of default rules in
more detail, we briefly review those elements of functional logic languages
and Curry that are necessary to understand the contents of this paper.
More details can be found in recent surveys on
functional logic programming \cite{AntoyHanus10CACM,Hanus13}
and in the language report \cite{Hanus16Curry}.

Curry is a declarative multi-paradigm language
combining in a seamless way features from functional,
logic, and concurrent programming (concurrency is
irrelevant as our work goes, hence it is ignored in this paper).
The syntax of Curry is close to Haskell \cite{PeytonJones03Haskell},
i.e., type variables and names of defined operations
usually start with lowercase letters and the names of type
and data constructors start with an uppercase letter.
\code{$\alpha \to \beta$} denotes the type of all functions
mapping elements of type $\alpha$ into elements of type $\beta$
(where $\beta$ can also be a functional type, i.e.,
functional types are ``curried''),
and the application of an operation $f$
to an argument $e$ is denoted by juxtaposition (``$f~e$'').
In addition to Haskell, Curry allows
\emph{free} (\emph{logic}) \emph{variables}
in conditions and right-hand sides of rules
and expressions evaluated by an interpreter.
Moreover, the patterns of a defining rule can be non-linear,
i.e., they might contain multiple occurrences of some variable,
which is an abbreviation for equalities between these occurrences.

\begin{example}\label{ex-conc}
The following simple program shows the functional and logic features
of Curry. It defines an operation \ccode{++} to concatenate two lists,
which is identical to the Haskell encoding.
The second operation, \code{dup}, returns some list element having at least
two occurrences:\footnote{Note that Curry requires the explicit declaration
of free variables, as \code{x} in the rule of \code{dup},
to ensure checkable redundancy.}
\begin{curry}
(++) :: [a] -> [a] -> [a]
[]     ++ ys = ys
(x:xs) ++ ys = x : (xs ++ ys)$\listline$
dup :: [a] -> a
dup xs | xs == _$\,$++$\,$[x]$\,$++$\,$_$\,$++$\,$[x]$\,$++$\,$_
       = x
       where$\;$x$\;$free
\end{curry}
\end{example}
Operation applications can contain free variables.
They are evaluated lazily where free variables as demanded arguments
are non-deterministically instantiated.
Hence, the condition of the rule defining \code{dup}
is solved by instantiating \code{x} and the anonymous free variables \ccode{\us}.
This evaluation method corresponds to narrowing \cite{Slagle74,Reddy85},
but Curry narrows with possibly non-most-general unifiers
to ensure the optimality of computations \cite{AntoyEchahedHanus00JACM}.

Note that \code{dup} is a \emph{non-deterministic operation}
since it might deliver more than one result for a given argument,
e.g., the evaluation of \code{dup$\,$[1,2,2,1]} yields the values
\code{1} and \code{2}.
Non-deterministic operations, which are
interpreted as mappings from values into sets of values \cite{GonzalezEtAl99},
are an important feature
of contemporary functional logic languages.
Hence, there is also a predefined \emph{choice} operation:
\label{ex:nondetchoice}
\begin{curry}
x ? _  =  x
_ ? y  =  y
\end{curry}
Thus, the expression \ccode{0$~$?$~$1} evaluates to \code{0} and \code{1}
with the value non-determin\-istically chosen.

Some operations can be defined more easily and directly using
\emph{functional patterns} \cite{AntoyHanus05LOPSTR}.
A functional pattern is a pattern occurring in an argument
of the left-hand side of a rule containing defined operations
(and not only data constructors and variables).
Such a pattern abbreviates the set of all standard patterns to which the
functional pattern can be evaluated (by narrowing).
For instance, we can rewrite the definition of \code{dup} as
\begin{curry}
dup (_++[x]++_++[x]++_) = x
\end{curry}
Functional patterns are a powerful feature to express arbitrary selections
in tree structures, e.g., in XML documents \cite{Hanus11ICLP}.
Details about their semantics and a constructive implementation
of functional patterns by a demand-driven unification procedure
can be found in \cite{AntoyHanus05LOPSTR}.

\emph{Set functions} \cite{AntoyHanus09} allow
the encapsulation of non-determi\-nistic computations
in a strategy-independent manner.
For each defined operation $f$, $f\funset$ denotes
the corresponding set function.
$f\funset$ encapsulates the non-determinism
caused by evaluating $f$ except for the non-determinism caused
by evaluating the arguments to which $f$ is applied.
For instance, consider the operation \code{decOrInc} defined by
\begin{curry}
decOrInc x = (x-1) ? (x+1)
\end{curry}
Then \ccode{decOrInc$\funset$ 3} evaluates to (an abstract representation of)
the set $\{\code{2},\code{4}\}$, i.e., the non-determinism caused by
\code{decOrInc} is encapsulated into a set.
However, \ccode{decOrInc$\funset$ (2\,?\,5)} evaluates to
two different sets $\{\code{1},\code{3}\}$ and $\{\code{4},\code{6}\}$
due to its non-deterministic argument, i.e.,
the non-determinism caused by the argument is not encapsulated.
This property is desirable and essential to define and implement
default rules by a transformational approach, as shown in
Sect.~\ref{sec:transsem}.
In the following section, we discuss default rules and
their intended semantics.

\section{Default Rules: Concept and Informal Semantics}
\label{sec:concept}

Default rules are often used in both functional and
logic programming.
In languages in which rules are applied in textual order,
such as Haskell and Prolog,
loosely speaking every rule is a default rule of all the preceding rules.
For instance, the following standard Haskell function
takes two lists and returns the list of corresponding pairs,
where excess elements of a longer list are discarded:
\begin{curry}
zip (x:xs) (y:ys) = (x,y) : zip xs ys
zip _      _      = []
\end{curry}
The second rule is applied only if the first rule is not applicable,
i.e., if one of the argument lists is empty.
We can avoid the consideration of rule orderings by replacing
the second rule with rules for the patterns not matching the first rule:
\begin{curry}
zip (x:xs) (y:ys) = (x,y) : zip xs ys
zip (_:_)  []     = []
zip []     _      = []
\end{curry}
In general, this coding is cumbersome since the number of
additional rules increases if the patterns of the first rule
are more complex (e.g., we need three additional rules
for the operation \code{zip3} combining three lists).
Moreover, this coding might be impossible in conjunction
with some functional patterns, as in the first rule of \code{isSet} above.
Some functional patterns conceptually denote an infinite set
of standard patterns
(e.g., \code{[x,x]}, \code{[x,\us,x]}, \code{[\us,x,\us,x]}, \ldots)
and the complement of this set is infinite too.

In Prolog, one often uses the ``cut'' operator to
implement the behavior of default rules.
For instance, \code{zip} can be defined as a Prolog predicate
as follows:
\begin{prolog}
zip([X|Xs],[Y|Ys],[(X,Y)|Zs]) :- !, zip(Xs,Ys,Zs).
zip(_,_,[]).
\end{prolog}
Although this definition behaves as intended for instantiated lists,
the completeness of logic programming is destroyed by the cut operator.
For instance, the goal \code{zip([],[],[])} is provable,
but Prolog does not compute the answer \code{\{Xs=[],Ys=[],} \code{Zs=[]\}}
for the goal \code{zip(Xs,Ys,Zs)}.

These examples show that neither the functional style nor
the logic style of default rules is suitable for functional logic programming.
The functional style, based on textual order, curtails non-determinism.
The logic style, based on the \emph{cut} operator, destroys the completeness
of some computations.
Thus, a new concept of default rules
is required for functional logic programming
if we want to keep the strong properties of the base language,
in particular, a simple-to-use non-determinism and
the completeness of logic-oriented evaluations.
Before presenting the exact definition of default rules,
we introduce them informally and discuss their intended semantics.

We intend to extend a ``standard'' operation definition by one default rule.
Hence, an operation definition with a default rule has the following form
($\ol{o_k}$ denotes a sequence of objects $o_1 \ldots o_k$):\footnote{%
We consider only conditional rules since an unconditional rule
can be regarded as a conditional rule with condition \code{True}.}
\begin{curry}
$f$ $\ol{t^1_k}$ | $c_1$ = $e_1$
$\vdots$
$f$ $\ol{t^n_k}$ | $c_n$ = $e_n\listline$
$f$'default $\ol{t^{n+1}_k}$ | $c_{n+1}$ = $e_{n+1}$
\end{curry}
We call the first $n$ rules \emph{standard rules} and the final rule
the \emph{default rule} of $f$.
Informally, the default rule is applied only if no standard
rule is applicable, where a rule is applicable if the pattern matches
and the condition is satisfied.
Hence, an expression $e = f~\ol{s_k}$,
where $\ol{s_k}$ are expressions, is evaluated as follows:

\begin{enumerate}
\item
  The arguments $\ol{s_k}$ are evaluated enough to determine
  whether a standard rule of $f$ is applicable, i.e.,
  whether there exists a standard rule whose left-hand side
  matches the evaluated $e$ and the condition is satisfied 
  (i.e., evaluable to \code{True}).
\item
  If a standard rule is applicable, it is applied; otherwise
  the default rule is applied.
\item
  If some argument is non-deterministic,
  the previous points apply independently
  for each non-deterministic choice of the combination of arguments.
  In particular, if an argument is a free variable, it is
  non-deterministically instantiated so that every
  potentially applicable rule can be used.
\end{enumerate}
As usual in a non-strict language like Curry,
arguments of an operation application
are evaluated as they are demanded by the operation's
pattern matching and condition. However, any non-determinism or failure
during argument evaluation is not passed inside the
condition evaluation.  A precise definition of  ``inside''
is in \cite[Def. 3]{AntoyHanus09}.
This behavior is quite similar to set functions to
encapsulate internal non-determinism.
Therefore, we will exploit set functions to implement default rules.

Before discussing the advantages and implementation of default rules,
we explain and motivate the intended semantics of our proposal.
First, it should be noted that this concept distinguishes
non-determinism outside and inside a rule application.
This difference is irrelevant in purely functional programming but essential
in functional logic programming.

\begin{example}\label{ex:zipdefault}
Consider the operation \code{zip} defined with a default rule:
\begin{curry}
zip (x:xs) (y:ys) = (x,y) : zip xs ys
zip'default _ _   = []
\end{curry}
Since the standard rule is applicable to \code{zip$\,$[1]$\,$[2]},
the default rule is ignored so that this expression is solely reduced to
\code{(1,2):zip\,[]\,[]}.
Since the standard rule is not applicable to
\code{zip$\,$[]$\,$[]}, the default rule is applied and yields
the value \code{[]}.
Altogether, the only value of \code{zip$\,$[1]$\,$[2]} is \code{[(1,2)]}.
However, if some argument has more than one value,
we use the evaluation principle above for each combination.
Thus, the call \code{zip\,([1]\,?\,[])\,[2]} yields the two values
\code{[(1,2)]} and \code{[]}.
\end{example}
These considerations are even more relevant
if the evaluation of the condition might be non-deterministic,
as the following example shows.

\begin{example}\label{lookup-cond}
Consider an operation to look up values for keys in an association list:
\begin{curry}
lookup key assoc   | assoc == (_ ++ [(key,val)] ++ _)
                   = Just val
                   where val free
lookup'default _ _ = Nothing
\end{curry}
Note that the condition of the standard rule can be evaluated
in various ways. In particular, it can be evaluated (non-deterministically)
to \code{True} and \code{False} for a fixed association list and key.
Therefore, using if-then-else (or an \code{otherwise} branch
as in Haskell)
instead of the default rule might lead to unintended results.

If we evaluate \code{lookup$\;$2$\;$[(2,14),(3,17),(2,18)]},
the condition of the standard rule is satisfiable
so that the default rule is ignored.
Since the condition has the two solutions $\{\code{val}\mapsto 14\}$
and $\{\code{val}\mapsto 18\}$, we yield the values
\code{Just$\,$14} and \code{Just$\,$18}.
If we evaluate \code{lookup$\;$2$\;$[(3,17)]}, the condition of the
standard rule is not satisfiable but the default rule is applicable so that
we obtain the result \code{Nothing}.

On the other hand, non-deterministic arguments might trigger
different rules to be applied. Consider the expression
\code{lookup\,(2\,?\,3)\,[(3,17)]}. Since the non-determin\-ism
in the arguments leads to independent evaluations of the expressions
\code{lookup$\;$2} \code{[(3,17)]} and \code{lookup$\;$3$\;$[(3,17)]},
we obtain the results \code{Nothing} and \code{Just$\,$17}.

Similarly, free variables as arguments might lead to independent results
since free variables are equivalent to non-deterministic values
\cite{AntoyHanus06ICLP}.
For instance, the expression \code{lookup$\;$2$\;$xs}
yields the value \code{Just$\;$v} with
the binding $\{\code{xs}\mapsto \code{(2,v)\code{:}\_}\}$, but also
the  value \code{Nothing} with
the binding $\{\code{xs}\mapsto\code{[]}\}$ (as well as many other solutions).

The latter desirable property also has implications for the handling
of failures occurring when arguments are evaluated.
For instance, consider the expression \code{lookup$\;$2$\;$failed}
(where \code{failed} is a predefined operation which always fails
whenever it is evaluated).
Because the evaluation of the condition of the standard rule
demands the evaluation of \code{failed} and the subsequent failure
comes from ``outside'' the condition,
the entire expression evaluation fails instead of returning
the value \code{Nothing}.
This is motivated by the fact that we need the value of the
association list in order to check the satisfiability of the
condition and, thus, to decide the applicability of the standard rule,
but this value is not available.
\end{example}

\begin{example}\label{ex:isUnit}
To see why our design decision is reasonable,
consider the following contrived definition of an operation
that checks whether its argument is the unit value \code{()}
(which is the only value of the unit type):
\begin{curry}
isUnit x | x == () = True
isUnit'default _   = False
\end{curry}
In our proposal, the evaluation of \ccode{isUnit$\;$failed} fails.
In an alternative design (like Prolog's if-then-else construct),
one might skip any failure during condition checking and
proceed with the next rule. In this case, we would return
the value \code{False} for the expression \code{isUnit$\;$failed}.
This is quite disturbing since the (deterministic!) operation
\code{isUnit}, which has only one possible input value,
could return two values: \code{True} for the call \code{isUnit$\,$()}
and \code{False} for the call \code{isUnit$\;$failed}.
Moreover, if we call this operation with a free variable,
like \code{isUnit$\;$x}, we obtain the single binding
$\{\code{x}\mapsto\code{()}\}$ and value \code{True}
(since free variables are never bound to failures).
Thus, either our semantics would be incomplete for logic computations
or we compute too many values.
In order to get a consistent behavior,
we require that failures of arguments demanded for
condition checking lead to failures of evaluations.
\end{example}

\section{Examples}
\label{sec:examples}

To show the applicability and convenience of default rules for functional logic programming,
we sketch a few more examples in this section.

\begin{example}
Default rules are important in combination with functional patterns,
since functional patterns denote an infinite set of standard patterns
which often has no finite complement.
Consider again the operation \code{lookup} as introduced in
Example~\ref{lookup-cond}.
With functional patterns and default rules, this operation
can be conveniently defined:
\begin{curry}
lookup key (_ ++ [(key,val)] ++ _) = Just val
lookup'default _  _                = Nothing 
\end{curry}
\end{example}

\begin{example}
Functional patterns are also useful to check the deep structure of
arguments. In this case, default rules are useful to express in an easy manner
that the check is not successful.
For instance, consider an operation that checks whether
a string contains a float number (without an exponent but with
an optional minus sign).
With functional patterns and default rules, the definition
of this predicate is easy:
\begin{curry}
isFloat (("-" ? "") ++ n1 ++ "." ++ n2)
        | (all isDigit n1 && all isDigit n2) = True
isFloat'default _ = False
\end{curry}
\end{example}

\begin{example}\label{ex:queens}
In the classical $n$-queens puzzle, one must place $n$ queens
on a chess board so that no queen can attack another queen.
This can be solved by computing some permutation
of the list \code{[1..$n$]}, where the $i$-th element denotes the row
of the queen placed in column $i$, and check whether this permutation
is a safe placement so that no queen can attack another in a diagonal.
The latter property can easily be expressed with functional patterns
and default rules where the non-default rule fails
on a non-safe placement:
\begin{curry}
safeDiag (_++[x]++zs++[y]++_) | abs (x-y) == length zs + 1 = failed
safeDiag'default xs = xs
\end{curry}
Hence, a solution can be obtained by computing a safe permutation:
\begin{curry}
queens n = safeDiag (permute [1..n])
\end{curry}
This example shows that default rules are a convenient way
to express negation-as-failure from logic programming.
\end{example}

\begin{example}\label{ex:colormap}
This programming pattern can also be applied to solve the map coloring problem.
Our map consists of the states of the Pacific Northwest
and a list of adjacent states:
\begin{curry}
data State = WA | OR | ID | BC$\listline$
adjacent = [(WA,OR),(WA,ID),(WA,BC),(OR,ID),(ID,BC)]
\end{curry}
Furthermore, we define the available colors and an operation
that associates (non-deterministically) some color to a state:
\begin{curry}
data Color = Red | Green | Blue$\listline$
color x = (x, Red ? Green ? Blue)
\end{curry}
A map coloring can be computed by an operation \code{solve}
that takes the information about potential colorings
and adjacent states as arguments, i.e., we compute correct colorings
by evaluating the initial expression
\begin{curry}
solve (map color [WA,OR,ID,BC]) adjacent
\end{curry}
The operation \code{solve} fails on a coloring where
two states have an identical color and are adjacent,
otherwise it returns the coloring:
\begin{curry}
solve (_++[(s1,c)]++_++[(s2,c)]++_) (_++[(s1,s2)]++_) = failed
solve'default cs _ = cs
\end{curry}
Note that the compact definition of the standard rule of \code{solve}
exploits the ordering in the definition of \code{adjacent}.
For arbitrarily ordered adjacency lists, we have to extend the
standard rule as follows:
\begin{curry}
solve (_++[(s1,c)]++_++[(s2,c)]++_) (_++[(s1,s2)$\,$?$\,$(s2,s1)]++_)
  = failed
\end{curry}
\end{example}

\section{Transformational Semantics}
\label{sec:transsem}

In order to define a precise semantics of default rules,
one could extend an existing logic foundation of functional logic
programming (e.g., \cite{GonzalezEtAl99}) to include a meaning
of default rules.
This approach has been partially done in \cite{LopezSanchez04}
but without considering the different sources of non-determinism
(inside vs.{} outside) which is important for our intended semantics,
as discussed in Sect.~\ref{sec:concept}.
Fortunately, the semantic aspects of these issues have already
been discussed in the context of encapsulated search
\cite{AntoyHanus09,ChristiansenHanusReckSeidel13PPDP}
so that we can put our proposal on these foundations.
Hence, we do not develop a new logic foundation of functional logic programming
with default rules, but we provide a transformational semantics,
i.e., we specify the meaning of default rules by a transformation
into existing constructs of functional logic programming.

We start the description of our transformational approach by
explaining the translation of the default rule for \code{zip}.
A default rule is applied only if no
standard rule is applicable (because the rule's pattern does not
match the argument or the rule's condition is not satisfiable).
Hence, we translate a default rule into a regular rule by adding
the condition that no other rule is applicable.
For this purpose, we generate from the original standard rules
a set of ``test applicability only'' rules where the right-hand side
is replaced by a constant (here: the unit value \ccode{()}).
Thus, the single standard rule of \code{zip} produces 
the following new rule:
\begin{curry}
zip'TEST (x:xs) (y:ys) = ()
\end{curry}
Now we have to add to the default rule the condition that \code{zip'TEST}
is not applicable.
Since we are interested in the failure of attempts
to apply \code{zip'TEST} to the actual argument,
we have to check that this application has no value.
Furthermore, non-determinism and failures in the evaluation of
actual arguments must be distinguished from similar outcomes
caused by the evaluation of the condition.

All these requirements call for the encapsulation of a search for
values of \code{zip'TEST} where ``inside'' and ``outside''
non-determinism are distinguished and handled differently.
Fortunately, set functions \cite{AntoyHanus09}
(as sketched in Sect.~\ref{sec:flp})
provide an appropriate solution to this problem.
Since set functions have a strategy-independent denotational semantics
\cite{ChristiansenHanusReckSeidel13PPDP},
we will use them to specify and implement default rules.
Using set functions, one could translate the default rule into
\begin{curry}
zip xs ys | isEmpty (zip'TEST$\funset$ xs ys) = []
\end{curry}
Hence, this rule can be applied only if all attempts
to apply the standard rule fail.
To complete our example, we add this translated default rule
as a further alternative to the standard rule so that
we obtain the transformed program
\begin{curry}
zip'TEST (x:xs) (y:ys) = ()$\listline$
zip (x:xs) (y:ys) = (x,y) : zip xs ys
zip xs ys | isEmpty (zip'TEST$\funset$ xs ys) = []
\end{curry}
Thanks to the logic features of Curry, one can also use
this definition to generate appropriate argument values for \code{zip}.
For instance, if we evaluate the equation \code{zip$\;$xs$\;$ys$\;$==$\;$[]}
with the Curry implementation KiCS2 \cite{BrasselHanusPeemoellerReck11},
the search space is finite and computes, among others, the solution
\code{\{xs=[]\}}.

Unfortunately, this scheme does not yield the best code
to ensure optimal computations. To understand the potential problem,
consider the following operation:\label{ex-indseq}
\begin{curry}
f 0 1 = 1
f _ 2 = 2
\end{curry}
Intuitively, the best strategy to evaluate a call to \code{f}
starts with a case distinction on the second argument,
since its value determines which rule to apply.
If the value is 1, and only in this case,
the strategy checks the first argument,
since its value determines whether to apply the first rule.
A formal characterization of operations that allow this strategy
\cite{Antoy92ALP} and a discussion of the strategy itself
will be presented in Sect.~\ref{sec:replace}.
In this example, the pattern matching strategy is as follows:
\begin{enumerate}
\item Evaluate the second argument (to head normal form).
\item If its value is \code{2}, apply the second rule.
\item If its value is \code{1}, evaluate the first argument
      and try to apply the first rule.
\item Otherwise, no rule is applicable.
\end{enumerate}
In particular, if \code{loop} denotes a non-terminating operation,
the call \code{f$\;$loop$\;$2} evaluates to \code{2}.
This is in contrast to Haskell \cite{PeytonJones03Haskell}
which performs pattern matching from left to right so that
Haskell loops on this call.
This strategy, which is optimal for the class of programs
referred to as \emph{inductively sequential} \cite{Antoy92ALP}
for which it is intended,
has been extended to functional logic computations
(\emph{needed narrowing} \cite{AntoyEchahedHanus00JACM})
and to overlapping rules \cite{Antoy97ALP}
in order to cover general functional logic programs.

Now consider the following default rule for \code{f}:
\begin{curry}
f'default _ x = x
\end{curry}
If we apply our transformation scheme sketched above,
we obtain the following Curry program:
\begin{curry}
f'TEST 0 1 = ()
f'TEST _ 2 = ()$\listline$
f 0 1 = 1
f _ 2 = 2
f x y | isEmpty (f'TEST$\funset$ x y) = y
\end{curry}
As a result, the definition of \code{f} is no longer
inductively sequential since the left-hand sides of the first and
third rule overlap.
Since there is no argument demanded by all rules of \code{f},
the rules could be applied independently.
In fact, the Curry implementation KiCS2 \cite{BrasselHanusPeemoellerReck11}
loops on the call \code{f$\;$loop$\;$2} (since it tries to evaluate the first
argument in order to apply the first rule), whereas it yields
the result \code{2} without the default rule.

To avoid this undesirable behavior when adding default rules,
we could try to use the same strategy for the standard rules
and the test in the default rule.
This can be done by translating the original standard rules
into an auxiliary operation and redefining the original operation
into one that either applies the standard rules or the default rules.
For our example, we transform the definition of \code{f}
(with the default rule) into the following functions:
\begin{curry}
f'TEST 0 1 = ()
f'TEST _ 2 = ()$\listline$
f'INIT 0 1 = 1
f'INIT _ 2 = 2$\listline$
f'DFLT x y | isEmpty (f'TEST$\funset$ x y) = y$\listline$
f x y = f'INIT x y ? f'DFLT x y
\end{curry}
Now, both \code{f'TEST} and \code{f'INIT} are inductively sequential
so that the optimal needed narrowing strategy can be applied,
and \code{f} simply denotes a choice (without an argument evaluation)
between two expressions that are evaluated optimally.
Observe that at most one of these expressions is reducible.
As a result, the Curry implementation KiCS2 evaluates \code{f$\;$loop$\;$2}
to \code{2} and does not run into a loop.

The overall transformation of default rules can be
described by the following scheme (its simplicity is advantageous
to obtain a comprehensible definition of the semantics of default rules).
The operation definition
\begin{curry}
$f$ $\ol{t^1_k}$ | $c_1$ = $e_1$
$\vdots$
$f$ $\ol{t^n_k}$ | $c_n$ = $e_n$
$f$'default $\ol{t^{n+1}_k}$ | $c_{n+1}$ = $e_{n+1}$
\end{curry}
is transformed into
(where \code{$f$'TEST}, \code{$f$'INIT}, \code{$f$'DFLT} are
new operation identifiers):
\begin{curry}
$f$'TEST $\ol{t^1_k}$ | $c_1$ = ()
$\vdots$
$f$'TEST $\ol{t^n_k}$ | $c_n$ = ()

$f$'INIT $\ol{t^1_k}$ | $c_1$ = $e_1$
$\vdots$
$f$'INIT $\ol{t^n_k}$ | $c_n$ = $e_n$

$f$'DFLT $\ol{t^{n+1}_k}$ | isEmpty ($f$'TEST$\funset$ $\ol{t^{n+1}_k}$) && $c_{n+1}$ = $e_{n+1}$

$f$ $\ol{x_k}$ = $f$'INIT $\ol{x_k}$ ? $f$'DFLT $\ol{x_k}$
\end{curry}
Note that the patterns and conditions of the original rules
are not changed. Hence, this transformation is also compatible
with other advanced features of Curry, like functional patterns,
``as'' patterns, non-linear patterns, local declarations, etc.
Furthermore, if an efficient strategy exists for the 
original standard rules, the same strategy can be applied
in the presence of default rules.
This property can be formally stated as follows:
\begin{proposition}
\label{sequentiality}
Let $\Rc$ be a program without default rules, 
and $\Rc'$ be the same program except that
default rules are added to some operations of $\Rc$.
If $\Rc$ is overlapping inductively sequential, so is $\Rc'$.
\end{proposition}
\begin{proof}
Let $f$ be an operation of $\Rc$.
The only interesting case is when a default rule of $f$ is in $\Rc'$.
Operation $f$ of $\Rc$ produces four different operations of $\Rc'$:
$f$, $f\code{'DFLT}$, $f\code{'INIT}$, and $f\code{'TEST}$.
The first two are overlapping inductively sequential
since they are defined by a single rule.
The last two are overlapping inductively sequential when $f$ of $\Rc$ is
overlapping inductively sequential since they have the same definitional
tree as $f$ modulo a renaming of symbols.
\end{proof}
The above proposition could be tightened a little
when operation $f$ is non-overlap\-ping.
In this case three of the four operations produced by the
transformation are non-overlapping as well.
Prop.\,\ref{sequentiality} is important for the efficiency of computations.
In overlapping inductively sequential systems, needed redexes
exist and can be easily and efficiently computed~\cite{Antoy97ALP}.
If the original system has a
strategy that reduces only needed redexes, the transformed system has
a strategy that reduces only needed redexes.  
This ensures that optimal computations are preserved by the transformation
regardless of non-determinism.

This result is in contrast to Haskell (or Prolog),
where the concept of default rules is based on
a sequential testing of rules,
which might inhibit optimal evaluation and
prevent or limit non-determinism.
Hence, our concept of default rules is more powerful
than existing concepts in functional or logic programming
(see also Sect.~\ref{sec:related}).

We now relate values computed in the original system to those
computed in the transformed system and vice versa.
As expected, extending an operation with a default rule
preserves the values computed without the default rule.
\begin{proposition}
\label{completeness}
Let $\Rc$ be a program without default rules, 
and $\Rc'$ be the same program except that
default rules are added to some operations of $\Rc$.
If $e$ is an expression of $\Rc$ 
that evaluates to the value $t$ w.r.t.\ $\Rc$,
then $e$ evaluates to $t$ w.r.t.\ $\Rc'$.
\end{proposition}
\begin{proof}
Let $f\;\ol{t_k} \to u$ w.r.t.\ $\Rc$, for some expression $u$,
a step of the evaluation of $e$.  The only interesting case is when
a default rule of $f$ is in $\Rc'$.
By the definitions of $f$ and $f\code{'INIT}$ in $\Rc'$,
$f\;\ol{t_k} \to f\code{'INIT}\;\ol{t_k} \to u$ w.r.t.\ $\Rc'$.
A trivial induction on the length of the evaluation of $e$
completes the proof.
\end{proof}
The converse of Prop.~\ref{completeness} does not hold because
$\Rc'$ typically computes more values than $\Rc$---that is the reason
why there are default rules.
The following statement relates values computed in $\Rc'$ to
values computed in $\Rc$.
\begin{proposition}
\label{soundness}
Let $\Rc$ be a program without default rules, 
and $\Rc'$ be the same program except that
default rules are added to some operations of $\Rc$.
If $e$ is an expression of $\Rc$ 
that evaluates to the value $t$ w.r.t.\ $\Rc'$,
then either $e$ evaluates to $t$ w.r.t.\ $\Rc$ or some default
rule of $\Rc'$ is applied in $e \tostar t$ in $\Rc'$.
\end{proposition}
\begin{proof}
Let $A$ denote an evaluation $e \tostar t$ in $\Rc'$
that never applies default rules.
For any operation $f$ of $\Rc$,
the steps of $A$ are of two kinds: (1)
$f\;\ol{t_k} \to f\code{'INIT}\;\ol{t_k}$
or (2) $f\code{'INIT}\;\ol{t_k} \to t'$, for
some expressions $\ol{t_k}$ and $t'$.
If we remove from $A$ the steps of kind (1) and
replace $f\code{'INIT}$ with $f$, we obtain an evaluation
of $e$ to $t$ in $\Rc$.
\end{proof}
In Curry, by design, the textual order of the rules is irrelevant.
A default rule is a constructive alternative to a certain kind of failure.
For these reasons, a single default rule, as opposed to
multiple default rules without any order,
is conceptually simpler and adequate in practical situations.
Nevertheless, a default rule of an operation $f$
may invoke an auxiliary operation with multiple ordinary rules, thus,
producing the same behavior of multiple default rules of $f$.

\section{Implementation}
\label{sec:impl}

The implementation of default rules for Curry based
on the transformational approach
is available as a preprocessor.
The preprocessor is integrated into the compilation
chain of the Curry systems
PAKCS \cite{Hanus16PAKCS} and KiCS2 \cite{BrasselHanusPeemoellerReck11}.
In some future version of Curry, one could also add a specific
syntax for default rules and transform them in the front end
of the Curry system.

The transformation scheme shown in the previous section
is mainly intended to specify the precise meaning of default rules
(similarly to the specification of the meaning of guards
in Haskell \cite{PeytonJones03Haskell}).
Although this transformation scheme leads to
a reasonably efficient implementation,
the actual implementation can be improved in various ways.
In the following, we present two approaches to improve
the implementation of default rules.

\subsection{Avoiding Duplicated Condition Checking}
\label{sec:nodupimpl}

Our transformation scheme for default rules
generates from a set of standard rules
the auxiliary operations \code{$f$'TEST} and \code{$f$'INIT}.
\code{$f$'TEST} is used in the condition of the translated default rule
to check the applicability of a standard rule, whereas \code{$f$'INIT}
actually applies a standard rule.
Since both alternatives (standard rules or default rule)
are eventually tried for application, the pattern matching and
condition checking of some standard rule might be duplicated.
For instance, if a standard rule is applicable to some call
and the same call matches the pattern of the default rule,
it might be tried twice: (1) the standard rule is applied
by \code{$f$'INIT}, and (2) its pattern and condition is tested
by \code{$f$'TEST} in order to test the \mbox{(non-)emptiness} of
the set of all results. Although the amount of duplicated
work is difficult to assess accurately due to Curry's lazy evaluation strategy
(e.g., to check the non-emptiness in the condition of \code{$f$'DFLT},
it suffices to compute at most one element of the set),
there is some risk for operationally complex conditions or
patterns, e.g., functional patterns.

This kind of duplicated work can be avoided by a more sophisticated
transformation scheme
where the common parts of the definitions of \code{$f$'TEST}
and \code{$f$'INIT} are joined into a single operation.
This operation first tests the application of a standard rule
and, in case of a successful test, returns a continuation
to proceed with the corresponding rule.
For instance, consider the rules for \code{zip} presented in
Example~\ref{ex:zipdefault}.
The operations \code{zip'TEST} and \code{zip'INIT}
generated by our first transformation scheme can be joined
into a single operation \code{zip'TESTC} by the following transformation:
\begin{curry}
zip'TESTC (x:xs) (y:ys) = \_$~$-> (x,y) : zip xs ys

zip'DFLT _ _ = []

zip xs ys = let cs = zip'TESTC$\funset$ xs ys
             in if isEmpty cs then zip'DFLT xs ys
                              else (chooseValue cs) ()
\end{curry}
Now, the standard rule is translated into a rule
for the new operation \code{zip'TESTC} where the rule's right-hand side
is encapsulated into a lambda abstraction to avoid its immediate
evaluation if this rule is applied.
The actual implementation of \code{zip} first checks whether the set of all
such lambda abstractions is empty. If this is the case,
the standard rule is not applicable so that the default rule is applied.
Otherwise, we continue with the right-hand sides of
all applicable standard rules collected as lambda abstractions
in the set \code{cs}.\footnote{%
The operation \code{chooseValue}
non-deterministically chooses some value of the given set.}

The general transformation scheme to obtain this
behavior is defined as follows.
An operation definition of the form
\begin{curry}
$f$ $\ol{t^1_k}$ | $c_1$ = $e_1$
$\vdots$
$f$ $\ol{t^n_k}$ | $c_n$ = $e_n$
$f$'default $\ol{t^{n+1}_k}$ | $c_{n+1}$ = $e_{n+1}$
\end{curry}
is transformed into:
\begin{curry}
$f$'TESTC $\ol{t^1_k}$ | $c_1$ = \_$~$-> $e_1$
$\vdots$
$f$'TESTC $\ol{t^n_k}$ | $c_n$ = \_$~$-> $e_n$

$f$'DFLT $\ol{t^{n+1}_k}$ | $c_{n+1}$ = $e_{n+1}$

$f$ $\ol{x_k}$ = let cs = $f$'TESTC$\funset$ $\ol{x_k}$
        in if isEmpty cs then $f$'DFLT $\ol{x_k}$
                         else (chooseValue cs) ()
\end{curry}
Obviously, this modified scheme avoids the potentially duplicated
condition checking in standard rules, but it is more sophisticated
since it requires the handling of sets of continuations.
Depending on the implementation of set functions,
this might be impossible if the values are operations.
If the results computed by set functions
are actually sets (and not multi-sets), this scheme
cannot be applied since sets require an equality operation
on elements in order to eliminate duplicated elements.

Fortunately, this scheme is applicable with PAKCS \cite{Hanus16PAKCS},
which computes multi-sets as results of set functions so that
it does not require equality on elements.
Thus, we compare the run times of both schemes
for some of the operations shown above which contain
complex applicability conditions (functional patterns).
All benchmarks were executed on a Linux machine (Debian Jessie)
with an Intel Core i7-4790 (3.60Ghz) processor and 8GB of memory.
Figure~\ref{nodup-benchmark} shows the run times (in seconds)
to evaluate some operations with both schemes.
These benchmarks indicate that the new scheme might yield
a reasonable performance gain, although this clearly depends
on the particular example.
A further alternative transformation scheme is discussed
in the following section.

\begin{figure}[h]
\vspace*{-2ex}
\begin{center}
\begin{tabular}{lccccc}
\hline
System: & \multicolumn{5}{c}{PAKCS 1.14.0 \cite{Hanus16PAKCS}} \\
\hline
Operation:  & \code{isSet} & \code{isSet}   & \code{lookup} & \code{lookup}         & \code{queens} \\
Arguments:  & [1..1000]    & [1000,1..1000] & 5001              & 5000              & 6             \\
            &              &                & [(1,.)..(5000,.)] & [(1,.)..(5000,.)] &               \\
\hline
Sect.~\ref{sec:transsem}:  & 7.04 & 4.78    & 3.56              & 3.54              & 0.23 \\
\hline
Sect.~\ref{sec:nodupimpl}: & 2.27 & 2.28    & 1.81              & 3.57              & 0.23 \\
\hline
\end{tabular}
\end{center}
\vspace*{-2ex}
\caption{\label{nodup-benchmark}
Performance comparison of different transformation schemes.
}
\end{figure}

\subsection{Transforming Default Rules into Standard Rules}
\label{sec:replace}

In some situations, the behavior of a default rule can be 
provided by a set of standard rules.
Almost universally, standard rules are more efficient.
An example of this situation is provided with the operation \code{zip}.
In Example~\ref{ex:zipdefault} this operation
is defined with a default rule.
A definition using standard rules is shown 
at the beginning of Sect.~\ref{sec:concept}.
The input/output relations of the two definitions are identical.
In this section, we introduce a few concepts to describe how to obtain,
under sufficient conditions,
a set of standard rules that behave as a default rule.

The programs considered in this section are
constructor-based \cite{ODonnell77}
(the extension to functional patterns is discussed later).
Thus, there are disjoint sets of \emph{operation} symbols,
denoted by $f,g,\ldots$,
and \emph{constructor} symbols, denoted by $c, d,\ldots$
An \emph{$f$-rooted pattern} is an expression of the form $f\,\ol{t_n}$
where $f$ is an operation symbol of arity $n$, each $t_i$
is an expression consisting of variables and/or constructor symbols only,
and $f\,\ol{t_n}$ is \emph{linear},
i.e., there are no repeated occurrences of some variable.
A \emph{pattern} is an $f$-rooted pattern for some operation $f$.
A pattern is \emph{ground} if it does not contain any variable.
A \emph{program rule} has the form $l = r$ where the left-hand side $l$
is a pattern (the extension to conditional rules is discussed later).
Given a redex $t$ and a step $t \to u$, $u$ is called a
\emph{contractum} (of $t$).
Although Curry allows non-linear patterns for
the convenience of the programmer, they are transformed into
linear ones through a simple syntactic transformation.

In the following, we first consider a specific class of programs,
called \emph{inductively sequential}, where the rules
of each operation can be organized in a \emph{definitional tree}
\cite{Antoy92ALP}.
\begin{definition}[Definitional tree]
\label{def:deftree}
The symbols \emph{rule}, \emph{exempt}, and \emph{branch}, 
appearing below, are uninterpreted functions for classifying
the nodes of a tree.
A \emph{partial definitional tree} with an $f$-rooted pattern $p$
is either a rule node $rule(p = r)$, an exempt node $exempt(p)$,
or a branch node $branch(p,x,\ol{\tree_k})$, where $x$ is a variable in $p$
(also called the \emph{inductive variable}),
$\{c_1,\ldots,c_k\}$ is the set of all the constructors of the type of $x$,
the substitution $\sigma_i$ maps $x$ to $c_i\,\ol{x_{a_i}}$
(where $\ol{x_{a_i}}$ are all fresh variables and $a_i$ is the arity of $c_i$),
and $\tree_i$ is a partial definitional tree with pattern $\sigma_i(p)$
(for $i=1,\ldots,k$).
A \emph{definitional tree} $\tree$ of an operation $f$
is a finite partial definitional tree with pattern $f\,\ol{x_n}$,
where $n$ is the arity of $f$ and $\ol{x_n}$ are pairwise different
variables, such that $\tree$ contains all and only the rules defining $f$
(up to variable renaming).
In this case, we call $f$ \emph{inductively sequential}.
\end{definition}
Definitional trees have a comprehensible graphical representation.
For instance, the definitional tree of the operation \ccode{++}
defined in Example~\ref{ex-conc} is shown in Fig.~\ref{fig:conctree}.
In this graphical representation, the pattern of each node is shown.
The root node is a \emph{branch} and its children are \emph{rule} nodes.
The inductive variable of the \emph{branch} is the left operand of
``\code{++}''. Referring to Def.~\ref{def:deftree},
$\sigma_1$ maps this variable to ``\code{[]}''
and $\sigma_2$ to $(x:xs)$.
For rule nodes, the right-hand side of the rule is shown
below the arrow. Exempt nodes are marked by the keyword \emph{exempt},
as shown in Figures~\ref{non-minimal} and~\ref{deftree}.
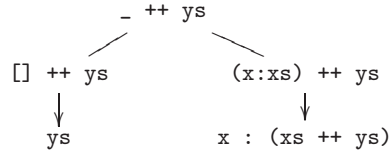
\begin{figure}[ht]
\hspace{\fill}
\xymatrix@C=-2pt@R=10pt{
& \code{\us{} ++ ys} \ar@{-}[dl] \ar@{-}[dr] \\
\code{[] ++ ys} \ar@{->}[d] & & 
    \code{(x:xs) ++ ys} \ar@{->}[d] \\
\code{ys} &  & \code{x : (xs ++ ys)} &
}
\hspace{\fill}
\caption{\label{fig:conctree}
A definitional tree of the operation \ccode{++}}.
\end{figure}

\noindent
For the sake of completeness, we sketch how definitional trees
are used by the evaluation strategy.
The details can be found in \cite{Antoy97ALP}.
We discuss how to compute a rewrite 
of an expression rooted by an operation.
More general cases are reduced to that.
It can be shown that any step so computed is needed.
Thus, let $t$ be an expression rooted by an operation $f$
and $\tree$ a definitional tree of $f$.
A traversal of $\tree$ finds the deepest node ${\cal N}$ in $\tree$ 
whose pattern $p$ matches $t$.  
Such a node, and pattern, exist for every $t$.
If ${\cal N}$ is a \emph{rule} node, then $t$ is a redex and is reduced.
If ${\cal N}$ is an \emph{exempt} node, then the computation
is aborted because $t$ has no value as in, e.g., \code{head []},
the head of an empty list.
If ${\cal N}$ is a \emph{branch} node, then the match of the inductive
variable of $p$ is an expression $t'$ rooted by some operation
and the strategy recursively seeks to compute a step of $t'$.

Before presenting our transformation,
we state an important property of definitional trees.
\begin{definition}[Mutually exclusive and exhaustive patterns]
\label{ex-ex}
Let $f$ be an operation symbol and $S$ a set of $f$-rooted patterns.
We say that the patterns of $S$ are \emph{mutually exclusive}
iff for any ground $f$-rooted pattern $p$, no two distinct patterns
of $S$ match $p$, and we say that the patterns of $S$ are 
\emph{exhaustive} iff for any ground $f$-rooted pattern $p$, 
there exists a pattern in $S$ that matches $p$.
\end{definition}
\begin{lemma}[Uniqueness]
\label{uniqueness}
Let $f$ be an operation defined by a set of standard rules.
If $\tree$ is a definitional tree of the rules of $f$,
then the patterns in the leaves of $\tree$ are exhaustive
and mutually exclusive.
\end{lemma}
\begin{proof}
Let $p$ be any ground $f$-rooted pattern,
$\pattern$ the pattern in a node $N$ of $\tree$,
and suppose that $\pattern$ matches $p$.
Initially, we show that
if $N$ is not a leaf of $\tree$,
there is exactly one child $N'$ of $N$
such the pattern $\pattern'$ of $N'$ matches $p$.
Let $x$ be the inductive variable of $\pattern$
and $q$ the subexpression of $p$ matched by $x$.
Since $p$ is ground and $q$ is a proper subexpression of $p$,
$q$ is rooted by some constructor symbol $c$.
Let $\{c_1,\ldots c_k\}$ be the set of all the constructors
of the type defining $c$ and 
let $a_i$ be the arity of $c_i$, for all appropriate $i$.
By Def.~\ref{def:deftree},
$N$ has $k$ children with patterns $\sigma_i(\pattern)$,
where $\sigma_i=\{x \mapsto c_i\,\ol x_{a_i}\}$ and
$\ol x_{a_i}$ is a fresh variable, for all appropriate $i$.
Hence, exactly one of these patterns matches $p$
since $c_i\,\ol x_{a_i}$ matches $q$ iff $c_i = c$.
Going back to the proposition's claim, since the pattern in the root
of $\tree$ matches $p$, by induction on the depth of $\tree$,
there is exactly one leaf whose pattern matches $p$.
\end{proof}
Inductive sequentiality is sufficient, but not necessary for
a set of exhaustive and mutually exclusive patterns.
We will later show a non-inductively sequential operation with 
exhaustive and mutually exclusive patterns.
Nevertheless, inductive sequentiality supports a constructive method
to transform default rules.
Since not every definitional tree is useful to define
our transformation, we first restrict the set of definitional trees.
\begin{definition}[Minimal definitional tree]
\label{def-minimal-tree}
A definitional tree is \emph{minimal} iff there is
some \emph{rule} node below any \emph{branch} node of the tree.
\end{definition}
For example, consider the operation \code{isEmpty} defined by the single rule
\begin{curry}
isEmpty [] = True
\end{curry}
Fig.~\ref{non-minimal} shows a non-minimal tree of 
the rules defining \code{isEmpty}.
The right child of the root is a \emph{branch} node
that has no \emph{rule} node below it.
In a minimal tree of the rules defining \code{isEmpty},
the right child would be an \emph{exempt} node.
\newcommand{\exempt}[1]{%
\vcenter{
\hbox{#1}
\hbox{(\emph{exempt})}
}}

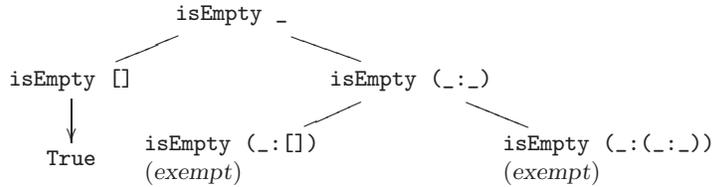
\begin{figure}[ht]
\hspace{\fill}
\xymatrix@C=-2pt@R=10pt{
& \code{isEmpty \us} \ar@{-}[dl] \ar@{-}[dr] \\
\code{isEmpty []} \ar@{->}[d] & & 
    \code{isEmpty (\us:\us)} \ar@{-}[dl] \ar@{-}[dr] \\
\code{True} & \exempt{\code{isEmpty (\us:[])}} & &
    \exempt{\code{isEmpty (\us:(\us:\us))}}
}
\hspace{\fill}
\caption{\label{non-minimal} A non-minimal
definitional tree of the operation \code{isEmpty}}.
\end{figure}

\noindent
We now investigate sufficient conditions for the equivalence
between an operation defined with a default rule and
an operation defined by standard rules only.

\begin{definition}[Replacement of a default rule]
\label{replacement}
Let $f$ be an operation defined by a set of standard rules and a
default rule $f\,\ol{x_k} = t$, 
where $\ol{x_k}$ are pairwise different variables
and $t$ some expression,
and let $\tree$ be a minimal definitional tree
of the standard rules of $f$.
Let $N_1, N_2, \ldots N_n$ be the exempt nodes of $\tree$,
$\ol{t^i_k}$ the pattern of node $N_i$ and 
$\sigma_i$ the substitution $\{\ol{x_k} \mapsto \ol{t^i_k}\}$,
for $1 \leqslant i \leqslant n$.
The following set of standard rules of $f$
is called a \emph{replacement} of the default rule of $f$:
\begin{equation}
\label{repl-rule}
\sigma_i(f~\ol{x_k} = t),    \hspace{12em} \mbox{for } 1 \leqslant i \leqslant n
\end{equation}
\end{definition}
Fig. \ref{deftree} shows a minimal definitional
tree of the single standard rule of operation \code{zip} defined 
at the beginning of Sect. \ref{sec:concept}.
The right-most leaf of this tree holds this rule.
Since this leaf is below both branch nodes,
the definitional tree is minimal according to Def.~\ref{def-minimal-tree}.
The remaining two leaves hold the patterns that
match all and only the combinations of arguments
to which the default rule would be applicable.
These patterns are more instantiated than that of the default rule,
but we will see that any expression reduced by these rules
does not need any additional evaluation with respect to the default
rule.

\begin{figure}[ht]
\hspace{\fill}
\xymatrix@C=-2pt@R=10pt{
& \code{zip \us{} \us{}} \ar@{-}[dl] \ar@{-}[dr] \\
\vbox{\hbox{\code{zip [] \us{}}}
      \hbox{~(\emph{exempt})}} & & \code{zip (x:xs) \us{}} \ar@{-}[dl] \ar@{-}[dr] \\
& \vbox{\hbox{\code{zip (x:xs) []}}
        \hbox{~~~~(\emph{exempt})}} & & \code{zip (x:xs) (y:ys)} \ar@{->}[d]\\
& & & \code{(x,y) : zip xs ys}
}
\hspace{\fill}
\caption{\label{deftree} A definitional tree of the standard rule of
operation \code{zip} defined in Sect. \protect{\ref{sec:concept}}}.
\end{figure}
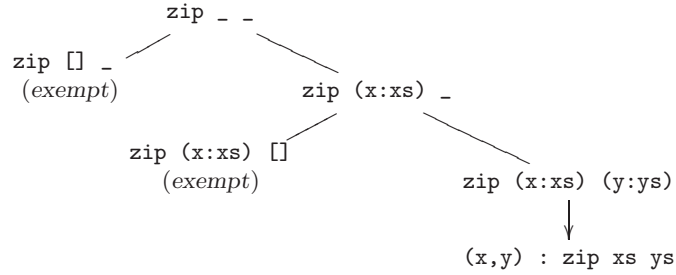

\begin{lemma}[Correctness]
\label{correctness}
Let $f$ be an operation defined by a set $S$ of standard rules and a
default rule $r$ of the form $f\,\ol{x_k} = t$, 
where each $x_i$ is a variable, for all appropriate $i$,
and some expression $t$,
and let $R$ be the replacement of $r$.
For any ground $f$-rooted pattern $p$,
$p$ is reduced at the root to some $q$ by the default rule $r$ iff
$p$ is reduced at the root to $q$ by some rule of $R$.
\end{lemma}
\begin{proof}
The proof is done in two steps.  First, we prove that
$p$ is reduced by $r$ iff $p$ is reduced by some rule of $R$.
Then we prove that the contracta by the two rules are the same.
By Lemma~\ref{uniqueness}, the patterns in the rules of $S \cup R$ are 
exhaustive and mutually exclusive.
Therefore, $p$ is reduced by $r$ if and only if 
$p$ is not reduced by any rule of $S$ if and only if
$p$ is reduced by some rule of $R$.
We now prove the equality of the contracta.
In the remainder of this proof, all the substitutions are
restricted to $\ol{x_k}$, the argument variables of $r$.
If $p$ is reduced by $r$ with some match $\sigma$, 
then $p=\sigma(f\,\ol{x_k})$ and $q=\sigma(t)$.
Pattern $p$ is also reduced by some rule of $R$ which,
by Def.~\ref{replacement}, is of the form $\sigma_i(r)$,
for some substitution $\sigma_i$.
Consequently, $p=\sigma'(\sigma_i(f\,\ol{x_k}))$ for some match $\sigma'$.
Since $p$ is ground, $\sigma=\sigma' \circ \sigma_i$.
Thus, the contractum of $p$ by the rule of $R$ is
$\sigma'(\sigma_i(t))=\sigma(t)=q$.
\end{proof}
In Def.~\ref{replacement}, the replacement of a default rule is
constructed for a minimal definitional tree.  The hypothesis of
minimality is not used in the proof of Lemma~\ref{correctness}.
The reason is that the lemma claims a property of $f$-rooted ground patterns.
During the execution of a program, the default rule
may be applied to some $f$-rooted expression $e$ that may neither be a
pattern nor ground.  The hypothesis of minimality ensures that,
in this case, no additional evaluation of $e$ is required when
a replacement rule is applied instead of the default rule.
This fact is counter intuitive and non-trivial
since the pattern of the default rule
matches any $f$-rooted expression, whereas the patterns in the replacement
rules do not, except in the degenerate case in which the set of
standard rules is empty.
However, a default rule is
applicable only if no standard rule is applicable.  Therefore,
expression $e$ must have been evaluated ``enough'' to determine
that no standard rule is applicable.
The following lemma shows that
this evaluation is just right for the application of a replacement rule.

\begin{lemma}[Evaluation]
\label{evaluation}
Let $e$ be an $f$-rooted expression reduced by the default rule of $f$
according to the transformational semantics of Sect.~\ref{sec:transsem}.
Let $\tree$ be a minimal definitional tree of (the standard rules of) $f$.
There exists an \emph{exempt} node of $\tree$ whose pattern matches $e$.
\end{lemma}
\begin{proof}
First note that the standard rules of $f$ and the rules
of \code{$f$'TEST}, as defined in Sect.~\ref{sec:transsem},
have identical left-hand sides.
Hence $\tree$ is also a minimal definitional tree of
the rules of \code{$f$'TEST}, which are used to check
the applicability of the default rule.

To prove the claim, we construct a path $N_0,N_1,\ldots N_p$
in the definitional tree $\tree$ of $f$ with the following invariant properties:
(a) the pattern $\pattern_i$ of each $N_i$ unifies with $e$, and,
(b) if the last node $N_p$ is a leaf of $\tree$, $N_p$ is an \emph{exempt} node.
Establishing the invariant:
$N_0$ is the root node of $\tree$.  By definition, its pattern
$\pattern_0$ is $f\,\ol{x_k}$, where $\ol{x_k}$ are fresh distinct variables.
Hence $\pattern_0$ unifies with $e$ so that invariant (a) holds.
Furthermore, if $N_0$ is a leaf of $\tree$,
then it cannot be a \emph{rule} node,
otherwise $e$ would never be reduced by a default rule.  
Hence $N_0$ is an \emph{exempt} node, i.e., invariant (b) holds.
Maintaining the invariant:
We assume that the invariant (a) holds for node $N_k$, 
for some $k \geqslant 0$.  If $N_k$ is a leaf of $\tree$,
then, as in the base case, $N_k$ must be an \emph{exempt} node.
Hence we assume $N_k$ is a \emph{branch} of $\tree$
and show that invariant (a) can be extended to some child $N_{k+1}$ of $N_k$.
Since $N_k$ is a \emph{branch} node, $e$ and $\pattern_k$ unify.
For each child $N'$ of $N_k$, let $\pattern'$ be the pattern of $N'$,
and let $v$ be the inductive variable of the branch node $N_k$.
By the definition of $\tree$, $\pattern'=\sigma'(\pattern_k)$,
where $\sigma'=\{v \mapsto c\,\ol{x_{a_c}}\}$,
$c$ is a constructor symbol of arity $a_c$,
and $x_i$ is a fresh variable for any appropriate $i$.
Let $\sigma$ be the match of $\pattern_k$ to $e$ and $t=\sigma(v)$.
If $t$ is a variable, then any child of $N_k$ satisfies invariant (a).
Otherwise, $t$ must be rooted by some constructor symbol, say $d$,
for the following reasons.
Because $\tree$ is minimal, there are one or more rule nodes below $N_k$.
The pattern in any of these rules is an instance of $\pattern_k$
that has some constructor symbol in the position matched by $v$.
Hence, unless $t$ is constructor-rooted,
it would be impossible to tell which, if any, of these rules reduces $e$,
hence it would be impossible to say whether $e$ must be reduced
by a standard rule or the default rule.
Hence, $N_{k+1}$ is the child in which $v$ is mapped to $d\,\ol{x_{a_d}}$
so that invariant (a) also holds for $N_{k+1}$.
\end{proof}
We define the replacement of a default rule by a set of standard rules
under four assumptions. We assess the significance of these assumptions below.

\smallskip\noindent{\bf Inductive sequentiality.}
The standard rules are inductively sequential.
This is a very mild requirement in practice.
For instance, every operation of the Curry \emph{Prelude},
except for the non-deterministic choice operator \ccode{?} shown
in Sect.~\ref{ex:nondetchoice}, is inductively sequential.
Non-inductively sequential operations are problematic to evaluate efficiently.
E.g., the following operation, adapted from
\cite[Prop.~II.2.2]{Berry76}, is defined
by rules that do not admit a definitional tree:
\begin{curry}
f False True  x     = ... 
f x     False True  = ...
f True  x     False = ...
\end{curry}
To apply $f$, the evaluation to constructor normal form
of two out of the three arguments is both necessary and sufficient.
No \emph{practical} way is known to determine
which these two arguments are without evaluating all three.
Furthermore, since the evaluation of an argument may not terminate, 
the three arguments must be evaluated concurrently
(but see \cite{AntoyMiddeldorp96TCS}).

\smallskip\noindent{\bf Most general pattern.}
We assumed in our transformation that
the pattern of the default rule is most general, i.e.,
the arguments of the operation are all variables.
Choosing the most general pattern keeps the statement of
Lemma~\ref{evaluation} simple and direct.
With this assumption, no extra evaluation of the arguments
is needed for the application of a replacement rule.
To relax this assumption, we can modify Def.~\ref{replacement}
as follows.
If the left-hand side of the default rule is $f\,\ol{u_k}$,
we look for a most general unifier, say $\sigma_i$,
of $\ol{u_k}$ and $\ol{t^i_k}$.
Then rule $\sigma_i(f\,\ol{u_k} \to t)$ is in the replacement
of the default rule iff such a $\sigma_i$ exists.

\smallskip\noindent{\bf Unconditional rules.}
Both standard rules and the default rule are unconditional.
Adding a condition to the default rule is straightforward,
similar to the transformation shown in Sect.~\ref{sec:transsem}.
The condition of a default rule is directly transferred to each replacement rule
by extending display (\ref{repl-rule}) in Def.~\ref{replacement}
with the condition.
By contrast, conditions in standard rules require some care.
With a modest loss of generality, assume that the standard rules
have a definitional tree where each leaf node has a conditional
rule of the form:
\begin{equation}
\label{conditional-rule}
f~\ol{t_k} \mid c = t
\end{equation}
where $c$ is a Boolean expression and $t$ is any expression.
Lemma~\ref{uniqueness} proves that if $p$ is any $f$-rooted ground pattern
matched by $f~\ol{t_k}$ no other standard rule matches $p$.
Hence, $p$ is reduced at the root by the default rule of $f$ iff
$c$ is not satisfied by $p$.
Therefore, we need the following rule in the replacement of the default rule
\begin{equation}
f~\ol{t_k} \mid \lnot c = t
\end{equation}
where $\lnot c$ denotes the ``negation'' of $c$,
i.e., the condition satisfied by all the patterns
matched by $f~\ol{t_k}$ that do not satisfy $c$.
In the spirit of functional logic programming,
$c$ is evaluated non-deterministically.
For example, consider an operation that takes a list of
colors, say \code{Red}, \code{Green} and \code{Blue}, and
removes all \code{Red} occurrences from the list:
\begin{curry}
data Color = Red | Green | Blue$\listline$
remred cs | cs == x++[Red]++y
          = remred (x++y)
          where x,y free
remred'default cs = cs
\end{curry}
The first rule is applied if there
exist \code{x} and \code{y} that satisfy the condition.
E.g., for the list \code{[Red,Green,Red,Blue]}
there are two such combinations of \code{x} and \code{y}.
Thus, the ``negation'' of this condition must negate
the existence of \emph{any} such \code{x} and \code{y}.
This can be automatically done according to the transformational semantics
presented in Sec.~\ref{sec:transsem}, but applied to a single rule.
This example's replacement of the default rule is shown below:
\begin{curry}
remred cs | isEmpty (remred'TEST$\funset$ cs) = cs
remred'TEST cs | _++[Red]++_ == cs = ()
\end{curry}

\smallskip\noindent{\bf Constructor patterns.}
The standard rules defining an operation have constructor patterns.
Curry also provides functional patterns,
presented in Sec.~\ref{sec:flp}.
Rules defined by functional patterns can be transformed into
ordinary rules \cite[Def.~4]{AntoyHanus05LOPSTR} by moving
the functional pattern matching into the condition of a rule.
Hence, the absence of functional patterns from our discussion
is not an intrinsic limitation.
Since functional patterns are quite expressive,
operations defined with functional patterns
often consist of a single program rule and a default rule
(as in all examples shown in in Sect.~\ref{sec:examples}).
For instance, the previous operation \code{remred}
can be defined with a functional pattern as follows:
\begin{curry}
remred (x++[Red]++y) = remred (x++y)
remred'default    cs = cs
\end{curry}
Hence, the improved transformation scheme presented
in~Sect.~\ref{sec:nodupimpl} is still useful and should be
applied in combination with the transformation shown in this section.

\section{Benchmarking}
\label{sec:bench}

To show the practical advantage of the transformation
described in the previous section, 
we evaluated a few simple operations defined in a typical
functional programming style with default rules.
For instance, the Boolean conjunction can be defined with a default rule:
\begin{curry}
and True  True  = True
and'default _ _ = False
\end{curry}
The replacement of the default rule consists of two rules
so that the transformation yields the following standard rules:
\begin{curry}
and True  True  = True
and True  False = False
and False _     = False
\end{curry}
Similarly, the computation of the last element of a list
can be defined with a default rule:
\begin{curry}
last [x] = x
last'default (_:xs) = last xs
\end{curry}
Our final example extracts all values in a list of optional
(``\code{Maybe}'') values:
\begin{curry}
catMaybes []             = []
catMaybes (Just x : xs)  = x : catMaybes xs
catMaybes'default (_:xs) = catMaybes xs
\end{curry}
With the introduction of default rules, the order of evaluation
may become more arbitrary, even though only needed steps are executed.
For example, in the first definition of operation \code{and}
both arguments must be evaluated, in any order,
for the application of the standard rule.
If the evaluation of one argument does not terminate and the other
one evaluates to \code{False},
the order in which the two arguments are evaluated becomes observable.
This situation is not directly related to the presence of a default rule.
There are two ``natural'' inductive definitions of operation \code{and},
one evaluates the first argument first, as in the second
definition of \code{and}, and another
evaluates the second argument first.
From the single standard rule of \code{and},
we cannot say which of the two definitions was intended.
If the default rule of operation \code{and} is replaced by a
set of standard rules, as per Sec.~\ref{sec:replace},
the resulting definition, which is inductively sequential,
will explicitly and arbitrarily encode which of the two arguments is
to be evaluated first.

As discussed earlier, functional logic computations execute narrowing
steps, i.e., steps in which some variable of an expression is
instantiated and the rule reducing the expression depends on the
instantiation of the variable.  For example, consider again the
\code{and} operation for its simplicity.
The evaluation of \code{and x True}, where \code{x} is a free variable,
narrows \code{x} to \code{True} to apply the standard rule and
narrows \code{x} to \code{False} to apply the default rule.
In a narrowing step, a variable is instantiated by the unification of
the expression being evaluated and the left-hand side of a rule.
This does not work with a default rule,
since the arguments in the left-hand side are themselves variables.
In particular, the transformational semantics of \code{and}
has no rule to unify \code{x} with \code{False}.
To obtain the intended behavior in narrowing steps
variables are instantiated by \emph{generators} \cite{AntoyHanus06ICLP}.
In the example being discussed, the Boolean generator is
\code{True ? False}.

Figure~\ref{performance} shows the run times (in seconds) to evaluate
the operations discussed in this section
with the different transformation schemes
(i.e., the scheme of Sect.~\ref{sec:transsem} and the
replacement of default rules presented in this section)
and different Curry implementations (where ``call size''
denotes the number of calls to \code{and} and the lengths
of the input lists for the other examples).
The benchmarks were executed on the same machine as the benchmarks
in Sect.~\ref{sec:nodupimpl}.
The results clearly indicate the advantage of replacing
default rules by standard rules, in particular for PAKCS,
which has a less sophisticated implementation of set functions
than KiCS2.
\begin{figure}[h]
\vspace*{-2ex}
\begin{center}
\begin{tabular}{lcccc}
\hline
System: & \multicolumn{4}{c}{PAKCS 1.14.0 \cite{Hanus16PAKCS}} \\
\hline
Operation / call size:  & \code{zip} / 1000 & \code{and} / 100000 & \code{last} / 2000 & \code{catMaybes} / 2000 \\
\hline
Sect.~\ref{sec:transsem}: & 3.66 & 8.46 & 2.53 & 2.45 \\
\hline
Sect.~\ref{sec:replace}:  & 0.01 & 0.25 & 0.01 & 0.01 \\
\hline
\hline
System: & \multicolumn{4}{c}{KiCS2 0.5.0 \cite{BrasselHanusPeemoellerReck11}} \\
\hline
Operation / call size:  & \code{zip} / $10^6$ & \code{and} / $10^6$ & \code{last} / $10^5$ & \code{catMaybes} / $10^6$ \\
\hline
Sect.~\ref{sec:transsem}: & 2.72 & 1.35 & 0.38 & 0.40 \\
\hline
Sect.~\ref{sec:replace}:  & 0.04 & 0.08 & 0.01 & 0.01 \\
\hline
\end{tabular}
\end{center}
\vspace*{-2ex}
\caption{\label{performance}
Performance comparison of different schemes for different compilers
for some operations discussed in this section. 
}
\end{figure}
%

\section{Related Work}
\label{sec:related}

In this section, we compare our proposal of default rules for Curry
with existing proposals for other rule-based languages.

The functional programming language Haskell \cite{PeytonJones03Haskell}
has no explicit concept of default rules.
Since Haskell applies the rules defining a function
sequentially from top to bottom,
it is a common practice in Haskell to write a ``catch all'' rule
as a final rule to avoid writing several nearly identical rules
(see example \code{zip} at the beginning of Sect.~\ref{sec:concept}).
Thus, our proposal for default rules increases the
similarities between Curry and Haskell.
However, our approach is more general, since it also supports
logic-oriented computations, and it is more powerful,
since it ensures optimal evaluation for inductively sequential
standard rules, in contrast to Haskell
(as shown in Sect.~\ref{sec:transsem}).

Since Haskell applies rules in a sequential manner,
it is also possible to define more than one default rule
for a function, e.g., where each rule has a different specificity.
This cannot be directly expressed with our default rules
where at most one default rule is allowed.
However, one can obtain the same behavior by introducing
a sequence of auxiliary operations where each operation has one default rule.

The logic programming language Prolog \cite{DeransartEdDbaliCervoni96}
is based on backtracking where the rules defining a predicate
are sequentially applied.
Similarly to Haskell, one can also define ``catch all'' rules
as the final rules of predicate definitions.
In order to avoid the unintended application of these rules,
one has to put ``cut'' operators in the preceding standard rules.
As already discussed in Sect.~\ref{sec:concept},
these cuts are only meaningful for instantiated arguments,
otherwise the completeness of logic programming might be destroyed.
Hence, this kind of default rules can be used only if the
predicate is called in a particular mode, in contrast to our approach.
The completeness for arbitrary modes might require
the addition of concepts from Curry into Prolog, like
the demand-driven instantiation of free variables.

Various encapsulation operators have been proposed for functional
logic programs \cite{BrasselHanusHuch04JFLP} to encapsulate
non-deterministic computations in some data structure.
Set functions \cite{AntoyHanus09} have been proposed as
a strategy-independent notion of encapsulating non-determinism
to deal with the interactions of laziness and encapsulation
(see \cite{BrasselHanusHuch04JFLP} for details).
One can also use set functions to distinguish successful and non-successful
computations, similarly to negation-as-failure in logic programming,
exploiting the possibility to check result sets for emptiness.
When encapsulated computations are nested and performed lazily,
it turns out that one has to track the encapsulation level
in order to obtain intended results, as discussed in
\cite{ChristiansenHanusReckSeidel13PPDP}.
Thus, it is not surprising that set functions and related operators
fit quite well to our proposal.
Actually, many explicit uses of set functions in functional logic
programming to implement negation-as-failure can be implicitly
and more tersely encoded with our concept of default rules,
as shown in Examples~\ref{ex:queens} and~\ref{ex:colormap}.

Default rules and negation-as-failure
have been also explored in \cite{LopezSanchez04,SanchezHernandez06}
for functional logic programs.
In these works, an operator, \code{fails}, is introduced to check
whether every reduction of an expression to a head-normal form
is not successful. \cite{LopezSanchez04} proposes the use of this operator
to define default rules for functional logic programming.
However, the authors propose a scheme where the default rule
is applied if no standard rule was able to compute a head normal form.
This is quite unusual and in contrast to functional programming
(and our proposal)
where default rules are applied if pattern matching and/or conditions
of standard rules fail, but the computations of the rules' right-hand sides
are not taken into account to decide whether a default rule should be
applied.
The same applies to an early proposal for default rules
in an eager functional logic language \cite{Moreno-Navarro94}.
Since the treatment of different sources of non-determinism
and their interaction were not explored at that time,
nested computations with failures are not considered by these works.
As a consequence, the operator \code{fails} might yield unintended
results if it is used in nested expressions.
For instance, if we use \code{fails} instead of set functions
to implement the operation \code{isUnit} defined
in Example~\ref{ex:isUnit}, the evaluation of \code{isUnit$\;$failed}
yields the value \code{False} in contrast to our intended semantics.

Finally, we proposed in \cite{AntoyHanus14}
to change Curry's rule selection strategy to a sequential one.
However, it turned out that this change has drawbacks
w.r.t.\ the evaluation strategy, since formerly optimal
reductions are no longer possible in particular cases.
For instance, consider the operation \code{f} defined
in Sect.~\ref{ex-indseq} and the call \code{f$\;$loop$\;$2}.
In a sequential rule selection strategy, one starts by testing
whether the first rule is applicable. Since both arguments
are demanded by this rule, one might evaluate them
from left to right (as done in the implementation \cite{AntoyHanus14})
so that this evaluation does not terminate.
This problem is avoided with our proposal which
returns \code{2} even in the presence of a default rule for \code{f}.
Moreover, the examples presented in \cite{AntoyHanus14}
can be expressed with default rules in a similar way.

\section{Conclusions}

We proposed a new concept of default rules for Curry.
Default rules are available in many rule-based languages,
but a sensible inclusion into a functional logic language
is demanding.
Therefore, we used advanced features for encapsulating search
to define and implement default rules.
Thanks to this approach, typical logic programming features, like
non-determinism and evaluating operations with unknown arguments, are
still applicable with our new semantics. This distinguishes our approach
from similar concepts in logic programming which simply cut alternatives.

Our approach can lead to more elegant and comprehensible
declarative programs, as shown by several examples in this paper.
Moreover, many uses of negation-as-failure, which are often
implemented in functional logic programs
by complex applications of encapsulation operators,
can easily be expressed with default rules.

Since encapsulated search is more costly than simple
pattern matching, we have also shown some opportunities
to implement default rules more efficiently.
In particular, if the standard rules are inductively sequential
and unconditional, one can replace the default rules
by a set of standard rules so that the usage of encapsulated
search can be completely avoided.

\paragraph{Acknowledgments.}
The authors are grateful to Sandra Dylus and the anonymous reviewers
for their suggestions to improve a previous version of this paper.
This material is based in part upon work supported 
by the National Science Foundation under Grant No.~1317249.


\end{document}